\documentclass[11pt]{article}

\usepackage[letterpaper,margin=1in]{geometry}
\usepackage[T1]{fontenc}
\usepackage[hyphens]{url}
\usepackage{microtype}
\usepackage{mathtools}
\usepackage{amssymb}
\usepackage{amsthm}
\usepackage[
  colorlinks=true,
  citecolor=blue,
  linkcolor=blue,
  urlcolor=blue
]{hyperref}
\usepackage[nameinlink,noabbrev]{cleveref}

\usepackage[
  style=authoryear,
  backend=biber,
  giveninits=true,
  uniquename=init,
  natbib=true,
  maxbibnames=99,
  minbibnames=99
]{biblatex}

\DeclareMathOperator{\supp}{support}
\DeclareMathOperator{\dist}{dist}

\theoremstyle{plain}
\newtheorem{theorem}{Theorem}

\newtheorem{corollary}[theorem]{Corollary}

\title{Two Observations on Metric Distortion and Condorcet Winning Sets}
\author{Jannik Peters\\
Shanghai University of Finance and Economics\\
\texttt{jannikpeters2512@gmail.com}}
\date{}

\begin{document}

\maketitle

\begin{abstract}
    In this research note we briefly connect two well-studied topics in computational social choice: metric distortion and the selection of undominated committees. In particular, we show that undominated committees are (in some sense) both necessary and sufficient to achieve bi-criteria metric distortion below $3$ \citep{BCC26a}. First, we show that any $\alpha$-undominated committee with $\alpha \le 0.5 - \Omega(1)$ has a bi-criteria metric distortion strictly of $3 - \Omega(1)$. In particular, this implies that a committee of size $5$ with distortion at most $2.7384$ exists. Secondly, we show that if a committee has a bi-criteria metric distortion strictly of $3 - \Omega(1)$, then it must also be $1 - \Omega(1)$-undominated. 
\end{abstract} 

\section{Introduction}
In this note, we show that two active lines of work in computational social choice are closely connected: metric distortion and the selection of undominated committees. 

Metric distortion tries to quantify the loss in social welfare due to voting rules not having access to full voter utilities, but only to ordinal rankings. In the metric distortion model, a set of voters and a set of candidates are lying in a common metric space. Voters prefer candidates that are close to them, and our goal is to select the candidate with the lowest social cost, that is, the lowest sum distances over all voters. However, we assume that we are only given the ordinal preferences of the voters, and not their actual distances. The distortion of a voting rule (using purely ordinal preferences), is the worst case approximation to the social cost, which this voting rule achieves. In this model, it is known that no deterministic voting rule can have a distortion of better than $3$ \citep{ABE+18a}, while there are voting rules achieving $3$ \citep{GHS20a, KiKe22a}. Recent work has started to investigate settings in which one can improve upon the lower bound of $3$. In particular, as shown by \citet{CRWW24b} randomized voting rules can achieve an (expected) distortion of strictly below $3$. This was further recently extended by \citet{CGRW26a} who showed that to get a distortion of less than $3$ a randomized voting rule with a constant size support is sufficient (with it being an open question as to how large the support of a voting rule needs to be to achieve distortion below $3$). Related to this, \citet{BCC26a} study a variant in which the distortion is not measured with regard to a single candidate, but for a set of candidates, with the distance of a voter to this set being the distance to their closest candidate in the set. One open question they leave in their work is whether one can always achieve a distortion of less than $3$ using constantly many candidates.

Another topic that has recently gained significant attention is that of undominated committees (and relatedly of Condorcet winning sets) \citep{ELS15a, CLP+24a}. For a fraction $\alpha \in [0,1]$ a set of candidates is called $\alpha$-undominated, if at most an $\alpha$ fraction of the voters can prefer a single candidate outside of the committee to everyone inside the committee. It is known that $\mathcal{O}(\frac{1}{k})$-undominated committees of size $k$ always exist \citep{CLP+24a, NSL26a}. An important question is for which values of $k$ a $\frac{1}{2}$-undominated committee, a so-called Condorcet winning committee exists (with such an existence currently being known for $k \ge 5$ \citep{NSL26a}). Currently all known ways of constructing a $1 - \Omega(1)$-undominated committee of a constant size are inherently non-constructive (using for instance Kakutani's fixed point theorem or von Neumann's minimax theorem), prompting \citeauthor{CLP+24a} to ask: ``are
there efficient and practical algorithms for finding Condorcet winning sets (or more generally,
$\alpha$-undominated sets)?''

\subsection{Our Contribution}
We show that these two problems are inherently connected. Namely, we show that any $0.5 - \Omega(1)$-undominated committee also has a bi-criteria metric distortion strictly of $3 - \Omega(1)$. As a consequence, since such committees exist for $k \ge 5$, one can achieve a bi-criteria metric distortion of strictly better than $3$ for $k \ge 5$. In particular, using the construction of \citet{NSL26a} one can achieve a distortion of $2.7384$ with a committee of size $5$.

Secondly, we show that any committee with a bi-criteria metric distortion of $3 - \Omega(1)$ it must also be $1 - \Omega(1)$-undominated. This also further extends to the support of randomized voting rules achieving a metric distortion of $3 - \Omega(1)$. Hence, if one wants to find either a candidate or a random voting rule with a distortion of $3 - \Omega(1)$, one necessarily needs to also construct an undominated committee. In particular, constructing such a set of constant size or a rule with constant size support in a simple or constructive 
manner would answer the open question of \citet{CLP+24a}.

\section{Model and Notation}
Throughout, we assume that we are given a set of voters $N = \{1, \dots, n\}$ and a set of candidates $C = \{c_1, \dots, c_m\}$. Each voter $i \in N$ possesses a complete ordinal preference ranking $\succ_i$ over $C$. Together $\mathcal{I} = (N, C, (\succ_i)_{i \in N})$ form an election instance. For a given set $W \subseteq C$ of candidates and candidate $c \notin W$ we denote by $M(W, c) \coloneqq \lvert \{i \in N \mid \exists c' \in W\colon c' \succ_i c \}\rvert$ the \emph{margin} of $W$ over $c$.

We study two questions in this paper: the selection of low metric distortion candidates and the selection of undominated committees.

\paragraph{Metric Distortion.}
Given an election instance $(N, C, (\succ_i)_{i \in N})$ we say that a function $d\colon (N \cup C) \times (N \cup C) \to \mathbb{R}_{\ge 0}$ is a \emph{compatible pseudo-metric} if
\begin{itemize}
    \item $d(x,x) = 0$ for all $x \in N \cup C$; 
    \item $d(x,y) = d(y,x)$ for all $x,y \in N \cup C$; 
    \item $d(x,y) \le d(x,z) + d(z,y)$ for all $x,y,z \in N \cup C$; 
    \item $d(i,c_1) \le d(i, c_2)$ if $c_1 \succ_i c_2$ for all $i \in N$ and $c_1, c_2 \in C$.
\end{itemize} For a given pseudo-metric $d$, voter $i \in N$ and set of candidates $C' \subseteq C$ we write $d(i,C') = \min_{c \in C'} d(i,c)$.

Given a pseudo-metric $d$ we say that a subset of candidates $C' \subseteq C$ has a distortion of $\lambda$ if 
\(
\sum_{i \in N} d(i,C') \le \lambda \sum_{i \in N} d(i, c)
\) for all $c \in C$. For a given election instance $\mathcal{I} = (N, C, (\succ_i)_{i \in N})$, we say a subset $C' \subseteq C$ has a metric distortion of $\dist(C')$ if the distortion of $C'$ for every compatible pseudo-metric $d$ (with $\mathcal{I}$) is $\dist(C')$.

\paragraph{Undominated Committees.}
Let $W\subseteq C$ be a subset of the candidates. For $\alpha > 0$ we say that $W$ is \emph{$\alpha$-undominated} if $M(W,c) \ge (1 - \alpha)n$ for all $c \notin W$.

\section{Main Results}

We begin with a positive connection between undominated committees and bi-criteria metric distortion. For this, we generalize the well-known connection between the metric distortion of a candidate and its maximum margin in a given election. 

\begin{theorem}[{\citep[Lemma~6]{ABE+18a}}]
    Let $c \in C$ be any candidate. Then 
    \[
    \dist(c) \le \max_{c' \in C \setminus \{c\}} \frac{2n}{M(c, c')} - 1.
    \]
\end{theorem}
In particular, this shows that Condorcet winners have a distortion of at most $3$ as $M(c, c') \ge \frac{n}{2}$ since they beat every other candidate in a pairwise comparison. Here, we show that this result generalizes to sets of candidates as well.
\begin{theorem}
    Let $W \subseteq C$ be any set of candidates. Then
    \[
    \dist(W) \le \max_{c' \in C \setminus W} \frac{2n}{M(W, c')} - 1.
    \]
    \label{thm:metric_margin}
\end{theorem}
\begin{proof}
    Let $d$ be any compatible pseudo-metric and let $c' \in C \setminus W$ be any candidate outside of $W$. Further, let
    $S \coloneqq \{i \in N \mid \exists w \in W\colon w \succ_i c'\}$ be the set of voters
    preferring some member of $W$ to $c'$. By definition,
    $\lvert S \rvert = M(W, c')$. For every voter $j \in S$, let $c_j \in W$ be their closest candidate in $W$. By the triangle inequality, we therefore have
    for every $j \in S$ that
    \[
        d(c', W) \le d(c', c_j) \le d(c', j) + d(j, c_j)
        \le 2d(j, c').
    \] Thus, by averaging over all $j \in S$, we get $d(c', W) \le \frac{2}{M(W, c')}\sum_{j \in S} d(j,c') \le \frac{2}{M(W, c')}\sum_{j \in N} d(j,c')$. 
    Now, for every voter $i \in S$ we have $d(i,W) \le d(i, c_i) \le d(i,c')$, while for every voter $i \in N \setminus S$ the
    triangle inequality gives $d(i, W) \le d(i, c') + d(c', W)$. Hence, we get
    \begin{align*}
        \sum_{i \in N} d(i, W)
        &\le \sum_{i \in N} d(i, c') + \lvert N \setminus S \rvert \cdot d(c', W)\\
        &\le \sum_{i \in N} d(i, c')
            + \left(n - M(W,c')\right) \frac{2}{M(W,c')} \sum_{i \in N} d(i, c')\\
        &= \left(\frac{2n}{M(W,c')} - 1\right) \sum_{i \in N} d(i, c'). \quad \qedhere
    \end{align*}
\end{proof}
As an immediate corollary we get an upper bound on the distortion of $\alpha$-undominated sets.
\begin{corollary}
Let $\alpha \ge 0$ and $W \subseteq C$ be an $\alpha$-undominated set of candidates. Then $\dist(W) \le 
    \frac{1+\alpha}{1-\alpha}$.
    \label{thm:dist_to_dom}
\end{corollary}
By using the result of \citet{NSL26a} that a $0.465008$-undominated set of size $5$ always exists, we also get that this set must have a distortion of at most $2.7384$.
\begin{corollary}
    In every instance with at least five candidates, there exists a set $W \subseteq C$ of size $5$ with a distortion of $\frac{1 + 0.465008}{1 - 0.465008} \le 2.7384$. 
\end{corollary}
Further, using the asymptotic bound of \citet{CLP+24a} we also get an asymptotic upper bound distortion achievable using a set of size $k$
\begin{corollary}
    In any instance with at least $k \ge 10$ candidates, there exists a set of $k$ candidates with distortion $1+\frac{19.6434}{(k-9.8217)}$.
    \label{cor:asymp}
\end{corollary}
\begin{proof}
    From \citet[Theorem~2]{CLP+24a} we know that a $\frac{9.8217}{k}$-undominated set of size $k$ always exists. Hence, by \Cref{thm:dist_to_dom} there always exists a set of size $k$ with distortion $\frac{1 + \frac{9.8217}{k}}{1 - \frac{9.8217}{k}} = 1+\frac{19.6434}{(k-9.8217)}$.
\end{proof}
\paragraph{Comparison with \citet{NSL25a}.}

In a preliminary version of their paper, \citet{NSL26a} also connected the bi-criteria distortion problem to their Lindahl equilibrium setting \citep{NSL25a}. Their bounds, however, are both worse and significantly more complicated than the bounds obtained here. In particular, they showed that one can achieve a distortion of $1 + \frac{16 \ln(4) \log_2(k)}{k}$ using a set of size $k$. This bound is asymptotically worse than the one of \Cref{cor:asymp} and is only below $3$ for $k \ge 68$. We further note that these results were not part of the final peer reviewed conference paper \citep{NSL26a}.

Finally, we prove our second connection: every set with a metric distortion of $3 - \Omega(1)$ must also be $1 - \Omega(1)$-undominated.
\begin{theorem}
   Let $C' \subseteq C$ be a set of candidates with metric distortion strictly less than $3 - \varepsilon$. Then $C'$ is also $\left(1 - \frac{\varepsilon}{2}\right)$-undominated.
\end{theorem}
\begin{proof}
    Let $C'$ be a committee that is not $\gamma \coloneqq 1 - \frac{\varepsilon}{2}$-undominated. 
    Thus, there exists a candidate $c \in C \setminus C'$ preferred to $C'$ by at least $\gamma \cdot n$ many voters. 
    \sloppy Now consider the pseudo-metric space with distance $d$ (following the general $(0,1,2,3)$-metric approach of \citet{ChRa22a}) in which for any two candidates $c_1, c_2 \in C$ we have $d(c_1, c_2) = 2$ and for any voter $i \in N$ and candidate $c' \in C$:
    \[
        d(i,c') = \begin{cases}
            1 \text{ if } &c' \succeq_i c\\ 
            3 \text{ if } &c \succ_i c'.
        \end{cases}
    \]
    For distinct voters $i,j \in N$, set $d(i,j)=2$, and set $d(x,x)=0$ for all $x$. We note that this indeed satisfies the triangle-inequality.

    We further get that by definition $\sum_{i \in N} d(i,c) = n$ while \[\sum_{i \in N} d(i,C') = |\{i \in N\mid \exists \hat{c} \in C' \colon  \hat{c} \succ_i c\}| + 3|\{i \in N \mid \forall \hat{c} \in C' \colon  c \succ_i \hat{c}\}| = (3\gamma + 1 - \gamma)n.\] Hence, the distortion of $C'$ is at least $2\gamma + 1 = 3-\varepsilon$.
\end{proof}
This also has a consequence for the ongoing research question of determining which lotteries can beat the distortion bound of $3$.
\begin{corollary}
    Let $\varepsilon > 0$. For any lottery $p$ with metric distortion better than $3-\varepsilon$ the support of $p$ is $1 - \frac{\varepsilon}{2}$-undominated.
\end{corollary}
\begin{proof}
    Let $p \in \Delta(C)$ be a lottery with expected distortion strictly less than $3-\varepsilon$.
    For every compatible pseudo-metric d and every candidate $c' \in C$,
    \[
    \sum_{i \in N} d(i,\supp(p))
    \le
    \sum_{i \in N}\sum_{c \in \supp(p)} p(c)d(i,c)
    <
    (3-\varepsilon)\sum_{i \in N} d(i,c').
    \]
\end{proof}
In particular, if one wants to find a simple lottery randomizing over two candidates only (an open question pointed out by \citet{CGRW26a} and \citet{CRWW24b}) and beating metric distortion $3$ one would as a consequence find a simple algorithm selecting a constant $1 - O(1)$ undominated committee of size two, an open question posed by \citet{CLP+24a}.

\section{Acknowledgments}
I thank Kangning Wang for feedback on this paper and for encouraging me to write it.

\printbibliography

@article{ABE+18a,
  author = {Elliot Anshelevich and Onkar Bhardwaj and Edith Elkind and John Postl and Piotr Skowron},
  doi = {10.1016/j.artint.2018.07.006},
  journal = {Artificial Intelligence},
  pages = {27--51},
  title = {Approximating Optimal Social Choice under Metric Preferences},
  volume = {264},
  year = {2018}
}

@inproceedings{BCC26a,
  author = {Kiarash Banishashem and Diptarka Chakraborty and Shayan {Chashm Jahan} and Iman Gholami and MohammadTaghi Hajiaghayi and Mohammad Mahdavi and Max Springer},
  booktitle = {Proceedings of the 14th International Conference on Learning Representations},
  note = {Forthcoming.},
  title = {Bi-Criteria Metric Distortion},
  url = {https://openreview.net/pdf?id=QBgHVmvN5S},
  year = {2026}
}

@misc{CGRW26a,
  archiveprefix = {arXiv},
  author = {Ziyi Cai and D. D. Gao and Prasanna Ramakrishnan and Kangning Wang},
  eprint = {2602.08871},
  primaryclass = {cs.GT},
  title = {Distortion of Metric Voting with Bounded Randomness},
  year = {2026}
}

@inproceedings{ChRa22a,
  author = {Moses Charikar and Prasanna Ramakrishnan},
  booktitle = {Proceedings of the 33rd Annual ACM-SIAM Symposium on Discrete Algorithms (SODA)},
  doi = {10.1137/1.9781611977073.116},
  pages = {2986--3004},
  title = {Metric Distortion Bounds for Randomized Social Choice},
  year = {2022}
}

@inproceedings{CLP+24a,
  author = {Moses Charikar and Alexandra Lassota and Prasanna Ramakrishnan and Adrian Vetta and Kangning Wang},
  booktitle = {Proceedings of the 57th Annual ACM Symposium on Theory of Computing (STOC)},
  doi = {10.1145/3717823.3718235},
  institution = {arXiv:2411.03390 [cs.GT]},
  pages = {1590--1601},
  title = {Six Candidates Suffice to Win a Voter Majority},
  year = {2025}
}

@article{CRWW24b,
  author = {Moses Charikar and Prasanna Ramakrishnan and Kangning Wang and Hongxun Wu},
  doi = {10.1145/3689625},
  journal = {Journal of the ACM},
  number = {6},
  pages = {1--33},
  title = {Breaking the Metric Voting Distortion Barrier},
  volume = {71},
  year = {2024}
}

@article{ELS15a,
  author = {Edith Elkind and J{\'e}r{\^o}me Lang and Abdallah Saffidine},
  doi = {10.1007/s00355-014-0853-4},
  journal = {Social Choice and Welfare},
  number = {3},
  pages = {493--517},
  title = {Condorcet winning sets},
  volume = {44},
  year = {2015}
}

@inproceedings{GHS20a,
  author = {Vasilis Gkatzelis and Daniel Halpern and Nisarg Shah},
  booktitle = {Proceedings of the 61st Symposium on Foundations of Computer Science (FOCS)},
  doi = {10.1109/FOCS46700.2020.00134},
  pages = {1427--1438},
  title = {Resolving the Optimal Metric Distortion Conjecture},
  year = {2020}
}

@inproceedings{KiKe22a,
  author = {Fatih Erdem Kizilkaya and David Kempe},
  booktitle = {Proceedings of the 31st International Joint Conference on Artificial Intelligence (IJCAI)},
  doi = {10.24963/ijcai.2022/50},
  pages = {349--355},
  title = {{P}lurality{V}eto: A Simple Voting Rule Achieving Optimal Metric Distortion},
  year = {2022}
}

@unpublished{NSL25a,
  author = {Th{\`a}nh Nguyen and Haoyu Song and Young-San Lin},
  note = {Unpublished manuscript. Available at \url{https://web.archive.org/web/20250529065531/https://web.ics.purdue.edu/~nguye161/few.pdf}},
  title = {A Few Good Choices},
  year = {2025}
}

@inproceedings{NSL26a,
  author = {Haoyu Song and Th{\'a}nh Nguyen and Young-San Lin},
  booktitle = {Proceedings of the 37th Annual ACM-SIAM Symposium on Discrete Algorithms (SODA)},
  doi = {10.1137/1.9781611978971.175},
  pages = {4861--4874},
  title = {A Few Good Choices},
  year = {2026}
}

\end{document}